\documentclass[submission,copyright,creativecommons]{eptcs}
\providecommand{\event}{QPL 2025} % Name of the event you are submitting to

\usepackage{graphicx} % Required for inserting images
\usepackage{amsmath, amsthm, amsfonts, amssymb}
\usepackage{braket}
\usepackage{framed}
\usepackage{hyperref}
\usepackage{comment}
\usepackage{enumitem}
\usepackage{diagbox}
\usepackage{makecell}
\usepackage{pdflscape}
\usepackage{mathtools}
\usepackage{pgf}
\usepackage{float}
\allowdisplaybreaks

\newtheorem{theorem}{Theorem}[section]

\newtheorem{conjecture}[theorem]{Conjecture}
\newtheorem{corollary}[theorem]{Corollary}
\newtheorem{lemma}[theorem]{Lemma}
\newtheorem*{theorem*}{Theorem}

\theoremstyle{definition}
\newtheorem{definition}[theorem]{Definition}

\newtheoremstyle{note}% style name
{\topsep}% above space
{\topsep}% below space
{}% body font
{}% indent amount
{\itshape}% head font
{.}% post head punctuation
{5pt plus 1pt minus 1pt}% post head space
{}% head spec

\theoremstyle{note}
\newtheorem{remark}[theorem]{Remark}

\newcommand{\N}{\mathbb{N}}
\newcommand{\Z}{\mathbb{Z}}
\newcommand{\Q}{\mathbb{Q}}
\newcommand{\R}{\mathbb{R}}

\newcommand{\A}{\mathbb{A}}
\newcommand{\B}{\mathbb{B}}

\newcommand{\M}{\mathcal{M}}

\newcommand{\BB}{\mathrm{B}}

\NewDocumentCommand\Bstate{D<>{\vec\lambda} O{\Phi}}{\ket{\BB(#1, #2)}}
\NewDocumentCommand\BstateNoKet{D<>{\vec\lambda} O{\Phi}}{\BB(#1, #2)}
\NewDocumentCommand\Bsimple{D<>{\lambda}}{\ket{\BB(#1)}}
\newcommand{\piovertwointerval}{\left[0, \piovertwo\right)}
\newcommand{\piovertwo}{\frac{\pi}{2}}
\newcommand{\pioverN}{\frac{\pi}{N}}
\newcommand{\modtwopi}{\mod{2\pi}}

\newcommand{\twolambdaszero}{interpolant}
\newcommand{\regular}{$N$-regular}

\DeclareMathOperator{\diag}{diag}

\usepackage{iftex}

\ifpdf
  \usepackage{underscore}         % Only needed if you use pdflatex.
  \usepackage[T1]{fontenc}        % Recommended with pdflatex
\else
  \usepackage{breakurl}           % Not needed if you use pdflatex only.
\fi

\title{A Classification Program for \\ Nonlocality Paradoxes of Three Qubits}
\author{Nadish de Silva \qquad\qquad  Santanil Jana
\qquad\qquad
Ming Yin \
\institute{Department of Mathematics, Simon Fraser University, Burnaby, BC, Canada}
}
\def\titlerunning{A Classification Program for Nonlocality Paradoxes of Three Qubits}
\def\authorrunning{N. de Silva, S. Jana, \& M. Yin}

\begin{document}
\maketitle

\begin{abstract}
Nonlocality is a quintessential signature of nonclassical behaviour and a resource for quantum advantages in communication and computation.  The paradoxical correlations witnessed by strong nonlocality undergird the standard probabilistic form of nonlocality and provide optimal advantages in numerous informational tasks.

Three-qubit systems are the simplest ones that admit strong nonlocality.  Abramsky et al.\ (\textit{TQC}, 2017) established the existence of an infinite family of three-qubit paradoxes, beyond the well-known GHZ paradox, which exhibited a novel conditional structure.

In this work, we introduce several new infinite families of three-qubit paradoxes and articulate a detailed roadmap towards the complete classification of all three-qubit nonlocality paradoxes.  In particular, we prove that our paradoxes exhaust all those satisfying reasonable regularity conditions.  We give an example of a highly exotic paradox and place constraints on the search for new exotic paradoxes.  We conjecture that all paradoxes must involve states from a one-parameter family and provide significant evidence in support of this conjecture.
\end{abstract}

\section{Introduction}  % 2-3 pages

Bell nonlocality \cite{Bell1964, Bell1966} describes how certain quantum correlations cannot be reproduced by a local hidden variable model.  %It crystallises a way in which quantum phenomena transcend classical explanation.  
Beyond its original foundational significance, it has more recently become the subject of intense research interest as a resource for realising advantages in quantum information-theoretic tasks.

The original arguments due to Bell and Clauser-Horne-Shimony-Holt \cite{CHSH} derive probabilistic inequalities satisfied by classical correlations, yet violated by certain quantum measurements on an entangled state. Greenberger et al.\ \cite{Greenberger_1990, Greenberger_1989} gave a stronger argument for nonlocality, involving the GHZ state, that does not rely on inequalities; rather it involves a logical argument that relies only on the \textit{possibilistic} data concerning outcomes of joint measurements.  A set of constraints is given that is satisfied by every empirically observable joint outcome, yet when taken together, form a paradox.  That is, no classical local hidden variable is consistent with the empirical observations of certain quantum systems.

Modern structural frameworks unify and generalise these concepts beyond the few previously known examples. Logical paradoxes were shown by Abramsky-Brandenburger \cite{Abramsky2011} in their sheaf-theoretic framework to unify diverse examples of \textit{strong nonlocality}.  Paradoxes undergird probabilistic nonlocality in that all nonlocal correlations are dilutions of paradoxes by classical randomness. The nonlocal fraction \cite{Barrett_2006, Elitzur_1992} measures nonlocality by quantifying the degree to which correlations are paradoxical.

Bell nonlocality now plays a pivotal role as a resource in quantum information, underpinning advantages in, e.g.\ nonlocal games \cite{Cleve_2004}, device-independent quantum cryptography \cite{colbeck2009quantum,Grasselli_2023}, and randomness generation \cite{acin2016certified}.  Nonlocality paradoxes also serve as resources in quantum computation.  For example, Anders-Browne \cite{Anders_2009} showed how access to measurements on GHZ states can promote a linear computer to classical universality; Raussendorf \cite{Raussendorf_2013} extended this to general measurement-based computation.  GHZ-type paradoxes also played a critical role in the result of Bravyi-Gosset-K\"onig on unconditional insimulability by shallow circuits \cite{Bravyi_2018}. Many results on nonlocality and quantum advantage show that advantages scale with the nonlocal fraction with paradoxes yielding deterministic advantages.

Motivated by the importance of nonlocality paradoxes in the study of quantum foundations and quantum advantage, we develop a program of exhaustively classifying all nonlocality paradoxes realisable by systems of three qubits.  This is the simplest system, in the sense of having minimal dimension, admitting strong nonlocality.  The first step in this direction was taken by Abramsky et al.\ \cite{Abramsky2017}, who showed that three-qubit nonlocality paradoxes must necessarily involve states LU-equivalent to ones in the SLOCC class of GHZ and local measurements in the equatorial $XY$-plane.  They further exhibited a family of such three-qubit paradoxes, indexed by even positive integers.  These paradoxes revealed a striking conditional structure not previously described. Raussendorf \cite{Raussendorf_2023} gave a topological interpretation of them.

A complete classification of three-qubit nonlocality paradoxes is a difficult mathematical question; we articulate a detailed roadmap towards its resolution and prove many key steps.  %Progress towards a full resolution has lead to the discovery of new paradoxes exhibiting highly exotic and novel structures; 
Further study of the structure of quantum paradoxes has potential applications across quantum information and computing.

\subsection{Summary of main results}
\begin{itemize}
    \item We study paradoxes up to a notion of equivalence based on local unitary equivalence and permutations of qubits.  We restrict ourselves to paradoxes that are minimal in the sense of not involving any extraneous measurements.  In Section \ref{section:2}, we define the family of \textit{interpolant states} to be those that interpolate between $\ket{\text{GHZ}}$ and $\ket{\text{Bell}} \otimes \ket{+}$ via a parameter $\lambda \in \piovertwointerval$.  All previously known paradoxes and the ones we demonstrate below involve these states.  We provide strong theoretical evidence in Section \ref{section:4.classification} for our Conjecture~\ref{conj:two-lambda-zero} that all three-qubit nonlocality paradoxes are equivalent to one using an interpolant state.
    \item We establish some constraints on the structure of impossible events for a fixed context that apply in any three-qubit model (see Lemmas \ref{lem:delta-facts} and \ref{lem:impossible-events}).  Any two impossible events must differ in at least two outcomes; thus, every context can have at most four impossible events.
    \item Any paradox using an interpolant state exhibits a conditional $\Z_2$-linear structure (Lemma \ref{lem:inconsistent}).  \textit{Maximal rank} paradoxes are those witnessed by summing \textit{all} linear equations to find $0=1$.
    \item We limit how many impossible events any choice of measurements on an interpolant state can yield overall and define a notion of \textit{maximally impossible} quantum scenarios of an interpolant state and measurement sets.  We focus on these quantum scenarios, following the intuition that paradoxes require many impossible events, and show how their measurement sets are constrained (Lemmas \ref{lem:M1=M2} and \ref{lem:delta_beta_restrict}).

    \item Theorem~\ref{thm:new_paradoxes} introduces several new infinite families of three-qubit nonlocality paradoxes.  One family is an extension of \cite[Theorem~8]{Abramsky2017}, which considered $\lambda_N = \piovertwo - \pioverN$ indexed by even $N\in \N$. Our construction, indexed by a tuple $(N,t)$, corresponds to a subset of the rationals and includes infinitely more paradoxes.  We give two more infinite families with significantly differing structure.  Theorem \ref{thm:inequivalent} shows that all our paradoxes are in distinct equivalence classes.   
    \item We give a partial classification of three-qubit nonlocality paradoxes up to equivalence in Theorem \ref{thm:partial_classification}.  We define \emph{\regular} quantum scenarios of a state and measurement sets that are maximally impossible, maximal rank, and use two measurements on the third qubit.  Our families of paradoxes exhaust all \emph{\regular} paradoxes up to equivalence.

    \item We give an example of a paradox that is not {\regular} and thus exhibits a very exotic structure.  We also study the consequences of relaxing the assumptions of regularity and indicate the most promising directions for completing the classification of quantum tripartite nonlocality paradoxes.

\end{itemize}

\paragraph{Outline.}Section \ref{section:2} summarises background material on nonlocality, and three-qubit paradoxes. Section~\ref{section:3} establishes our technical foundations by articulating the logical structure of reasonably regular paradoxes. In Section \ref{section:4}, we present our main results, including the presentation of new families of paradoxes in Subsection \ref{section:4.new_families}, a partial classification of three-qubit paradoxes in Subsection \ref{section:4.classification}, and a proposed roadmap towards a complete classification in Subsection \ref{section:4.complete_class}.  The appendices contain some detailed proofs and figures illustrating exotic paradoxes.

\section{Background} % 2-3 pages
\label{section:2}
\subsection{Strong nonlocality}
% This will be standard, from existing literature about Strong nonlocality.
We recall the basic definitions of the 
Abramsky-Brandenburger \cite{Abramsky2011} framework for nonlocality and contextuality (adapted to the setting of $n$-qubit nonlocality), and the notation of Abramsky et al.\ \cite{Abramsky2017}.

\textbf{Measurement scenarios} provide an abstract framework for describing experimental setups. Typically, a measurement scenario $(\mathcal{X}, \mathcal{O}, \mathcal{C})$ is defined by a set of measurement labels $\mathcal{X}$, a set of possible outcomes $\mathcal{O}$, and a cover $\mathcal{C}$ of $\mathcal{X}$, consisting of measurement contexts, which are maximal sets of measurements that can be jointly performed. Throughout this paper, we will set $\mathcal{O}$ to be $\Z_2 = \{0,1\}$ since our focus will be on quantum realisable paradoxes and projective measurements. 

% This is a more general definition - in the next section we should mention that we can restrict to equatorial measurements and label the measurements by angles.
A \textbf{Bell measurement scenario} $\M$ on $n$ qubits is an $n$-tuple of finite sets $(M_1, \ldots, M_n)$ with each $M_i$ a set of measurements on the $i$-th qubit. The \textbf{contexts} of $\M$ are given by $\mathcal{C} = \prod_{i=1}^n M_i$, i.e.\ choices of exactly one measurement per qubit. The \textbf{empirical model} $\varepsilon(\ket{\psi}, \M)$ yielded by a \textbf{quantum scenario} $(\ket{\psi}, \M)$ is a set of probability distributions $P_C : \mathcal{O}^n \to [0,1]$, indexed by contexts $C \in \mathcal{C}$, given by the Born rule.

A local measurement can be written in the form $E_{\theta, \varphi} := \sin\left(\theta\right)[\cos\left(\varphi\right) X + \sin\left(\varphi\right) Y] + \cos\left(\theta\right) Z$, where $\theta \in \piovertwointerval$ and $\varphi \in [0, 2\pi)$. This is a measurement with $+1$ eigenstate $\ket{\theta, \varphi} := \frac{1}{\sqrt{2}} (\cos\frac{\theta}{2}\ket{0} + e^{i\varphi}\sin\frac{\theta}{2} \ket{1})$ and $-1$ eigenstate $\ket{\pi - \theta, \varphi + \pi}$. Thus we can label a context by a tuple $(\vec\theta, \vec\varphi) :=$\linebreak $((\theta_1, \varphi_1),\ldots,(\theta_n, \varphi_n))$. We can then denote an \textbf{event} by $(\vec\theta, \vec\varphi) \rightarrow \vec o$, where $\vec o \in \mathcal{O}^n$ labels measurement outcomes (relabelling $+1, -1$ to $0, 1$ respectively).

An empirical model $\varepsilon(\ket{\psi}, \M)$ is \textbf{strongly nonlocal} if, given any \textbf{global assignment} $g: \bigsqcup_{i=1}^n M_i \rightarrow \mathcal{O}$ of outcomes to the measurements of $\M$, there exists a context $(\vec\theta, \vec\varphi)$ such that the event $(\vec\theta, \vec\varphi) \rightarrow (g(\theta_1, \varphi_1), \ldots, g(\theta_n, \varphi_n))$ is impossible. In quantum terms, if we denote the resulting eigenstate due to the event $(\vec\theta, \vec\varphi) \rightarrow \vec o$ by $\ket{(\vec\theta, \vec\varphi) \rightarrow \vec o}$, then the event is \textbf{impossible} if and only if $\braket{(\vec\theta, \vec\varphi) \rightarrow \vec o\, |\, \psi} = 0$. We refer to a quantum scenario $(\ket{\psi}, \M)$ such that $\varepsilon(\ket{\psi}, \M)$ is strongly nonlocal as a \textbf{(quantum nonlocality) paradox}.

\subsection{Three-qubit nonlocality paradoxes}
Brassard et al.\ \cite{Brassard_2005} proved that two-qubit quantum states do not exhibit strongly nonlocal behaviour. Thus, at least three qubits are required for strong nonlocality. Abramsky et al.\ \cite{Abramsky2017} showed the following:
\begin{theorem}[{\cite[Theorem 6]{Abramsky2017}}]
    A tripartite quantum state admitting strong nonlocality must be in the SLOCC class of the GHZ state and, in particular, must be \emph{balanced}. Moreover, any such strongly nonlocal behaviour can be witnessed using only \emph{equatorial measurements}.
\end{theorem}
%any three-qubit strongly nonlocal state must necessarily belong to the GHZ SLOCC (stochastic local operations and classical communication) class.
A \textbf{balanced} state has the form $\Bstate = \frac{1}{\sqrt{2}}(\ket{v_{\vec\lambda}} + e^{i \Phi} \ket{w_{\vec\lambda}})$,
where $\vec\lambda = (\lambda_1, \lambda_2, \lambda_3)$ with $\lambda_j \in \piovertwointerval$, $\Phi \in [0,2\pi )$, and $\ket{v_{\vec\lambda}} = \bigotimes_{i=1}^3 \ket{v_{\lambda_i}}$, $\ket{w_{\vec\lambda}} = \bigotimes_{i=1}^3 \ket{w_{\lambda_i}}$ with \begin{equation}
    \ket{v_{\lambda_j}} := \cos{\frac{\lambda_j}{2}} \ket{0} + \sin{\frac{\lambda_j}{2}} \ket{1}, \quad \ket{w_{\lambda_j}} := \sin{\frac{\lambda_j}{2}} \ket{0} + \cos{\frac{\lambda_j}{2}} \ket{1}.
\end{equation}
\textbf{Equatorial measurements} are of the form $E_\varphi := \cos\left(\varphi\right) X + \sin\left(\varphi\right) Y$ for $\varphi \in [0,2\pi)$. Since the only measurements that contribute to strong nonlocality are equatorial ones, we only need to consider measurement scenarios $\M = (M_1,M_2,M_3)$ where each set $M_i$ consists of measurement angles $\varphi \in [0, 2\pi)$. The following remark yields further simplifications.

\begin{remark}\label{rem:measurement}
    As discussed in Abramsky et al.\ \cite[Section 2.3]{Abramsky2017}, it suffices to only consider global assignments $g: \bigsqcup_{i=1}^n M_i \rightarrow \mathcal{O}$ with the property that, for any local equatorial measurement $E_\varphi$ on any qubit,
    \begin{equation}\label{eqn:nice_ga}
        g(\varphi) = g(\varphi + \pi) \oplus 1.
    \end{equation}
    That is, to show $\varepsilon(\ket\psi,\M)$ is strongly nonlocal, it suffices to demonstrate inconsistency of only those global assignments satisfying Eq.\ \ref{eqn:nice_ga}. Conversely, if  $\varepsilon(\ket\psi,\M)$ is not strongly nonlocal then there is a consistent global assignment that also satisfies Eq.\ \ref{eqn:nice_ga}. Thus, when constructing measurement scenarios, we can consider each set $M_i$ as only containing angles in the interval $[0, \pi)$.
\end{remark}

Letting $\vec\varphi = (\varphi_1, \varphi_2, \varphi_3)$, define $\ket{\vec\varphi} = \bigotimes_{i=1}^3 \ket{\varphi_i}$. We want to understand when the amplitude\linebreak $\braket{\vec\varphi|\BstateNoKet}$ is $0$, since this corresponds to the event $\vec\varphi \rightarrow (0,\ldots,0)$ being impossible.

\begin{lemma}[{\cite{Abramsky2017}}] \label{beta-snl}
    Let $\beta: \piovertwointerval \times [0, 2\pi) \rightarrow \R$ be defined as follows:
    \begin{equation}\label{eqn:beta}
        \beta(\lambda, \varphi) = \varphi - 2\arctan \left( \frac{\cos{\frac{\lambda}{2}} \sin{\varphi}}{\sin{\frac{\lambda}{2}} + \cos{\frac{\lambda}{2}} \cos{\varphi}} \right).
    \end{equation}
    Then \begin{equation}
    \braket{\vec\varphi|\BstateNoKet} = 0 \iff \sum_{i=1}^3\beta(\lambda_i, \varphi_i) \equiv \pi - \Phi \modtwopi.\label{eqn:imposs_}
\end{equation}
\end{lemma}

For the rest of the paper, the equivalence symbol $\equiv$ is used to denote when two quantities are equal modulo $2\pi$, unless otherwise specified.  We note that the formula for the function $\beta$ differs slightly from that of \cite[Section~5.3]{Abramsky2017}. The proof is straightforward and is therefore omitted. 
\begin{lemma}\label{lem:beta-facts}
    The following properties of the function $\beta$ can be easily verified. 
    \begin{enumerate}
        \item Modulo $2\pi$, for all fixed $\lambda \in \piovertwointerval$, $\beta(\lambda, \varphi)$ is strictly decreasing as a function of $\varphi$ on $(0, 2\pi)$ and is thus bijective on $[0, 2\pi)$.
        \item $\beta(0, \varphi) \equiv -\varphi $. 
        \item For all $\lambda \in \piovertwointerval$, $\beta(\lambda, \varphi) \equiv 0$ if and only if $\varphi \equiv 0 $, and $\beta(\lambda, \varphi) \equiv \pi $ if and only if $\varphi \equiv \pi $.
        \item $\beta(\lambda, \piovertwo) \equiv \lambda - \piovertwo $.
    \end{enumerate}
\end{lemma}

To avoid redundancy, we consider quantum scenarios $(\ket{\psi}, \M)$, $(\ket{\psi}', \M')$ to be \textbf{equivalent} if $U\ket{\psi} = \ket{\psi'}$ and $M \in M_i$ if and only if $UMU^\dag\in M_{P(i)}'$; here $P$ is some permutation of qubit labels and we have used shorthand notation to denote the rotation of an equatorial measurement, mod $\pi$, by some local unitary $U$.  Note that if $\vec\lambda$ has at least one zero entry, without loss of generality $\lambda_1 = 0$, then $(\Bstate, \M)$ is equivalent to $(\Bstate[0], \M')$ via the local unitary $\diag(1, e^{-i\Phi})\otimes I\otimes I$. 

For the rest of the paper, we mostly focus on a specific subclass of balanced states, defined as follows: 
\begin{definition}\label{def:twolambdaszero}
    An \textbf{interpolant state} $\Bsimple$ is a balanced state of the form $\ket{\text{B}((0,0,\lambda),0)}$.
\end{definition}
All examples of paradoxes found in \cite{Abramsky2017} and in Section \ref{section:4.new_families} of this paper involve interpolant states.  These states interpolate between $\ket{\text{B}(0)} = \ket{\text{GHZ}}$ and $\lim_{\lambda \to \piovertwo} \ket{\text{B}(\lambda)} = \ket{\text{Bell}} \otimes \ket{+}$.

\section{Technical preliminaries}\label{section:3}

In this section, we collect novel theoretical developments required for our main results.

Throughout this paper, we will assume that all paradoxes fulfil the following definition.
\begin{definition}
    A paradox $(\ket{\psi}, \M)$ is \textbf{minimal} if $(\ket{\psi}, \M')$ is not a paradox for any $\M'$ containing strictly fewer measurements than $\M$, i.e.\ $M_i' \subseteq M_i$ for all $i$ and at least one containment is strict.
\end{definition}

\subsection{Impossible events in quantum tripartite models}
\label{section:3.1}
% Put all the technical facts about delta, beta, impossible events in contexts etc here. Then in the main results we write a paragraph or two about how this is evidence for two lambdas = 0.

We define an additional function,
\begin{equation}
    \delta(\lambda, \varphi) := \beta(\lambda, \varphi + \pi) - \beta(\lambda, \varphi),
\end{equation}
which tells us how $\beta$ changes when we flip a measurement outcome on one qubit. Note that
\begin{equation}
    \delta(\lambda, \varphi) \equiv \pi - 2\arctan(\sin\varphi\tan\lambda) ,
\end{equation}
for which we use the fact that $\arctan u + \arctan v = \arctan\left(\frac{u+v}{1-uv}\right) + k\pi$ for some $k \in\{-1, 0, 1\}$.

\begin{lemma}\label{lem:delta-facts}~
    Let $\lambda \in \piovertwointerval$ and $\varphi \in [0, \pi)$. Then one can easily check the following properties of $\delta$.
    \begin{enumerate}
        \item  $\delta(\lambda, \varphi) \in (0, \pi] $.

        \item $\delta(\lambda, \varphi) \equiv \pi $ if and only if $\sin\varphi \tan\lambda = 0 \iff \lambda = 0 \text{ or } \varphi = 0$.

        \item $\delta(\lambda_1, \varphi_1) \pm \delta(\lambda_2, \varphi_2) \equiv 0  \iff \sin\varphi_1\tan\lambda_1 = \mp\sin\varphi_2\tan\lambda_2$.
        Further, by Lemma \ref{lem:delta-facts}.1, 
        \begin{equation}
            \delta(\lambda_1, \varphi_1) + \delta(\lambda_2, \varphi_2) \equiv 0  \iff \delta(\lambda_1, \varphi_1) = \delta(\lambda_2, \varphi_2) \equiv \pi.
        \end{equation}
    \end{enumerate}
\end{lemma}
Suppose that, considering measurement of the context $(A,B,C)$ on a state, the event $ABC \rightarrow abc$ is impossible.  By Lemma \ref{lem:delta-facts}.1, $\delta(\lambda, \varphi) \not\equiv 0\ \forall \lambda,\varphi$, so any other impossible event of the same context must have a tuple of outcomes whose Hamming distance from $abc$ is at least $2$. Therefore, the number of impossible events associated with any context is at most $4$.  

The following lemma is a consequence of Lemmas \ref{lem:delta-facts}.2 and \ref{lem:delta-facts}.3.

\begin{lemma}\label{lem:impossible-events}
    Fix a context $(A,B,C)$.
    \begin{enumerate}
        \item For a given $c \in \Z_2$, if one of the events $ABC \rightarrow 00c$, $ABC \rightarrow 11c$ is impossible, then the other event is impossible if and only if $\delta(\lambda_1, A) \equiv \delta(\lambda_2, B) \equiv \pi $.
        \item For a given $c \in \Z_2$, if one of the events $ABC \rightarrow 01c$, $ABC \rightarrow 10c$ is impossible, then the other event is impossible if and only $\delta(\lambda_1, A) \equiv \delta(\lambda_2, B)$.
    \end{enumerate}
    Similar results hold by permuting the qubits.
\end{lemma}

\subsection{Logical structure of three-qubit nonlocality paradoxes using interpolant states}
\label{section:3.2}
The family of three-qubit paradoxes introduced in \cite{Abramsky2017} differed from that of the GHZ state in that they involved not a single unsatisfiable system of linear equations over $\Z_2$ but two such systems.  These two systems were conditioned on the outcome of the third-qubit measurement and found to be inconsistent in both cases.  Here, we show that this conditional structure is a feature of every paradox involving an interpolant state $\Bsimple$.  We express this structure in a convenient form that we use in the remainder of the paper.

First we note that, by Lemma \ref{lem:beta-facts}.2, Eq.\ \ref{eqn:imposs_} has a much simpler form for interpolant states. Given a context $(A_j, B_k, C_l)$, the outcome $(a_j,b_k,c_l)$ is \emph{impossible} if and only if
\begin{equation}\label{eqn:imposs_events_interpolant}
    A_j + B_k \equiv \beta(\lambda, C_l) + c_l\delta(\lambda, C_l) + (1 \oplus a_j \oplus b_k) \pi,
\end{equation}
where $\oplus$ denotes addition modulo $2$. Note that $\beta(\lambda, C_l) + c_l\delta(\lambda, C_l) = \beta(\lambda, C_l + c_l \pi)$. By rearranging Eq.\ \ref{eqn:imposs_events_interpolant}, we see that a necessary condition for an event $A_jB_kC_l \rightarrow a_jb_kc_l$ to be impossible is that $A_j + B_k - \beta(\lambda, C_l) - c_l\delta(\lambda, C_l)$ must be a multiple of $\pi$.

Define $r_{jkl}(z) := \pi^{-1}[(A_j+B_k) - \beta(\lambda, C_l + z\pi)] \mod 2 \in [0,2)$.  By Lemma~\ref{beta-snl} and Lemma~\ref{lem:beta-facts}.2, the outcome $(a_j,b_k,c_l)$ is \emph{possible} if and only if \begin{equation}
     r_{jkl}(c_l) \notin \Z_2 \quad \text{or} \quad (a_j\oplus b_k) = r_{jkl}(c_l).
\end{equation}   

For $z \in \Z_2$, define the set of $\Z_2$-linear equations:\begin{equation}\label{eqn:Z2-linear-def}
    \Psi_{l}(z) := \{a_j \oplus b_k = r_{jkl}(z)\mid \forall j,k \text{ such that } r_{jkl}(z) \in \Z_2\}.
\end{equation}
A global assignment $g: \bigsqcup M_i \to \mathcal{O}$ is consistent with an empirical model of $(\Bsimple, \M)$ if and only if, for every $l$, $\Psi_{l}(z)$ is satisfied by $a_j = g(A_j),\ b_k = g(B_k),\ c_l = g(C_l)$.  We have thus shown the following.
\begin{lemma}\label{lem:inconsistent}
    The quantum scenario $(\Bsimple, \M)$, with $M_3 = \{C_0, \dots, C_{n-1}\}$, is a paradox if and only if for every $\vec z \in \Z_2^n$, the system of $\Z_2$-linear equations $\Psi(\vec z) := \cup_l \Psi_{l}(z_l)$ is inconsistent.
\end{lemma}

\subsection{Maximally impossible quantum scenarios}
\label{section:3.3}
Intuitively, an empirical model with more impossible events has a higher likelihood of witnessing a paradox compared to one with fewer impossible events. This motivates us to focus on empirical models in which the number of impossible events is maximised. We formalise this idea in the following definition.

\begin{definition}\label{def:max-imposs_}
    Let $(\Bsimple, \M )$ be a quantum scenario involving an interpolant state. Let $C\in M_3$, $z\in \Z_2$ and define: \begin{align}
        \A_z (C) &:= \{ A \in M_1 \mid \exists B \in M_2,\ a,b\in \Z_2: ABC \rightarrow abz \text{ is impossible} \} \label{eqn:max_imposs} \\
        \B_z (C) &:= \{ B \in M_2 \mid \exists A \in M_1,\ a,b\in \Z_2 : ABC \rightarrow abz \text{ is impossible} \}.
    \end{align}
    Also define $\A (C) := \A_0 (C) \cap \A_1 (C)$ and $\B (C) :=\B_0 (C) \cap \B_1 (C)$.
    Then $(\Bsimple , \M)$ %, and its empirical model, 
    is \textbf{maximally impossible} if for all $C\in M_3$, we have $\A (C) = M_1$ and $\B (C) = M_2$.  
\end{definition}

The definition of $\A (C)$ can be reformulated as follows: $A\in \A (C)$ if and only if $\exists B_1, B_2 \in M_2$ such that [$A+B_1 \equiv \beta(\lambda, C)$ or $\beta(\lambda, C)+ \pi$] and [$A+B_2 \equiv \beta(\lambda, C+\pi)$ or $ \beta(\lambda, C+\pi)+ \pi$]. The definition of $\B (C)$ can be reformulated in a similar way.

Now we prove some properties of maximally impossible quantum scenarios. First, observe that, by Eq.\ \ref{eqn:imposs_}, if $A \in \A_z(C)$, then there is exactly one measurement $B$ that satisfies the required property in Eq.\ \ref{eqn:max_imposs}, and likewise for measurements in $\B_z(C)$. This justifies our nomenclature.

\begin{lemma}\label{lem:M1=M2}
    Let $(\Bsimple, \M)$ be maximally impossible. Then $|M_1| = |M_2|$.
\end{lemma}
\begin{proof}
    Let $C\in M_3$ and fix an outcome $z \in \Z_2$. We define the function $f_1 \colon M_1 \to M_2$ as follows: $f_1(A) = B$, where $B$ is such that $ABC \rightarrow abz$ is an impossible event for some choice of $a, b\in \Z_2$. Since $(\Bsimple, \M)$ is maximally impossible, $f_1$ is well-defined and injective. By symmetry, we can define a function $f_2 \colon M_2 \to M_1$ such that $f_2$ is injective. This implies $|M_1| = |M_2|$. 
\end{proof}

\begin{remark}\label{rem:better-notation}
    Thus, given a maximally impossible quantum scenario $(\Bsimple, \M)$ with $|M_1| = |M_2| = N$ and $|M_3| = n$, there exists a (unique) function $K : \Z_N \times \Z_n \times \Z_2 \rightarrow \Z_N$ such that for all $(j, l, z) \in \Z_N \times \Z_n \times \Z_2$ there exists an impossible event $A_jB_{K(j, l, z)}C_l \rightarrow abz$ for some $a,b \in \Z_2$. Furthermore, $K(-,l,z)$ gives a bijection $\Z_N \leftrightarrow \Z_N$ for all $(l,z) \in \Z_n \times \Z_2$. We can then write the $\Z_2$-linear equations of Eq.\ \ref{eqn:Z2-linear-def} using the following notation, which emphasises the dependencies among the indices:
    \begin{equation}\label{eqn:Z2-linear-max-imposs}
        \Psi_{l}(z) = \{a_j \oplus b_{K(j,l,z)} = r(j,l,z)\mid j \in \Z_N\}.
    \end{equation}
\end{remark}

\begin{remark}\label{rem:Z2_rank}
    Suppose in particular that $M_3 = \{C_0, C_1\}$ (i.e.\ $n=2$).
    Let $(z_0, z_1) \in \Z_2^2$ and consider the linear system $\Psi(z_0, z_1) = \Psi_0 (z_0) \cup \Psi_1 (z_1)$. By the discussion above, each of the variables $a_j, b_k$ for $j,k \in \Z_N$ appears in exactly two of the $2N$ equations in $\Psi(z_0, z_1)$, so the sum of the LHS of the system is $0$. Write the system in matrix form, $\Gamma \vec a = \vec r$, where $\vec a = (a_0, \ldots, a_{N-1}, b_0, \ldots, b_{N-1})^T$ and $\vec r = \bigl({r(0,0,z_0)},\ \ldots,\ {r(N-1,0,z_0)},\ {r(0,1,z_1)},\ \ldots,\ {r(N-1,1,z_1)}\bigr)^T$. Then the kernel of the coefficient matrix $\Gamma$ is nontrivial and hence the rank of $\Gamma$ is at most $2N - 1$. The system is inconsistent if and only if the augmented matrix $(\Gamma \mid \vec r)$ has a higher rank than $\Gamma$.
\end{remark}

\begin{lemma}\label{lem:delta_beta_restrict}~
    Let $(\Bsimple, \M)$ be a maximally impossible quantum scenario. Let $|M_1| = |M_2|=N$ and $M_3=\{C_0, \dots, C_{n-1}\}$. Then $\delta (\lambda, C_l)$ and $\beta (\lambda, C_{l_1}) - \beta (\lambda, C_{l_2})$ are integer multiples of $\pioverN$ for all $l,l_1, l_2 \in \Z_n$.
\end{lemma}
\begin{proof}~
    Let $C_l \in M_3$ and $K$ be the function as defined in Remark \ref{rem:better-notation}. For every $A_j \in M_1$ and outcome $c_l = z \in \Z_2$, Eq.\ \ref{eqn:imposs_events_interpolant} is satisfied by $k = K(j,l,z)$. Therefore
    \begin{alignat}{2}
        A_j+B_{K(j,l,0)} &\equiv \beta(\lambda, C_l) + (1 \oplus a_j\oplus b_{K(j,l,0)})\pi \quad &&\text{for } j \in \Z_N \label{eq:eq1}  \\
        \text{and} \quad A_j+B_{K(j,l,1)} &\equiv \beta(\lambda, C_l) + \delta(\lambda, C_l) + (1 \oplus a_j\oplus b_{K(j,l,1)})\pi \quad &&\text{for } j \in \Z_N. \label{eq:eq2}
    \end{alignat}
    If we sum the $N$ equations from Eq.\ \ref{eq:eq1} and separately sum the $N$ equations from Eq.\ \ref{eq:eq2}, then subtract the former from the latter, we get $N \delta (\lambda, C_l) \equiv 0 \text{ or } \pi $. Note that the sum of the $B_{K(j,l,0)}$'s and the sum of the $B_{K(j,l,1)}$'s cancel out. So, $\delta (\lambda, C_l) \equiv s\pioverN $ for some $s\in \{1,\dots, N\}$. 
    
    Let $C_{l_1}, C_{l_2} \in M_3$. When $c_{l_1} = c_{l_2} = 0$, the maximally impossible property implies the following:  \begin{align}
        A_j+B_{K(j,l_1,0)} &\equiv \beta(\lambda, C_{l_1}) + (1 \oplus a_j\oplus b_{K(j,l_1,0)})\pi \quad \text{for } j \in \Z_N   \\
        \text{and} \quad A_j+B_{K(j,l_2,0)} &\equiv \beta(\lambda, C_{l_2}) + (1 \oplus a_j\oplus b_{K(j,l_2,0)})\pi \quad \text{for } j \in \Z_N. 
        \end{align}
    Similarly to above, if we add all the equations on the first line and subtract the equations on the second line, we get $N (\beta (\lambda, C_{l_1}) - \beta (\lambda, C_{l_2})) \equiv 0 \text{ or } \pi $. Therefore, $\beta (\lambda, C_{l_1}) - \beta (\lambda, C_{l_2}) \equiv t\pioverN $ for some $t\in \{0,\dots, N-1\}$.
\end{proof}

\section{Main results}
\label{section:4}
In this section, we introduce several families of paradoxes, including a natural extension of \cite[Theorem~8]{Abramsky2017}. We provide graphical visualisations of these paradoxes and establish that they are pairwise distinct up to equivalence. Next, we classify all $N$-regular (see Definition~\ref{def:Nregular}) paradoxes by showing that the previously constructed paradoxes exhaust all such cases. Additionally, we present a paradox that is not $N$-regular, suggesting the existence of exotic paradoxes with novel structures. Finally, we propose a conjecture that paradoxes must involve {\twolambdaszero} states, and provide supporting evidence.

\subsection{New families of paradoxes}
\label{section:4.new_families}
We first focus on maximally impossible quantum scenarios $(\Bsimple, \M)$ such that there are two measurements on the third qubit. Referring to Remark \ref{rem:Z2_rank}, we introduce an additional assumption that, for all $z_0, z_1 \in \Z_2$, the coefficient matrix $\Gamma$ of the corresponding linear system $\Psi(z_0, z_1)$ achieves its maximum possible rank of $2N - 1$. We shall use the term \textbf{maximal rank} to describe the quantum scenarios satisfying this assumption.

\begin{definition}\label{def:Nregular}
    The quantum scenario $(\Bsimple, \M)$ is \textbf{\regular} if it is maximally impossible and maximal rank, with $|M_1| = |M_2| = N$ and $M_3 = \{C_0, C_1 \}$. 
\end{definition}

\begin{lemma}\label{lem:opposite_parities}
    Let $(\Bsimple, \M)$ be {\regular}. For any $z_0, z_1 \in \Z_2$, the linear system $\Psi(z_0, z_1)$ is inconsistent if and only if 
    \begin{equation}
        \bigoplus_{j = 0}^{N-1} r(j,0,z_0) \neq \bigoplus_{j = 0}^{N-1} r(j,1,z_1).
    \end{equation}
\end{lemma}
\begin{proof}
    Since the quantum scenario is maximal rank, $\Psi(z_0, z_1)$ is inconsistent if and only if the rank of the corresponding augmented matrix is $2N$, which holds if and only if the sum of all the entries of $\vec r$ (the RHS of the matrix equation as defined in Remark \ref{rem:Z2_rank}) is $1$.
\end{proof}
\begin{remark}\label{rem:slight-generalisation}
    The `if' direction of Lemma \ref{lem:opposite_parities} holds even if we remove the property of maximal rank.
\end{remark}

\begin{corollary}\label{cor:opposite_parities}
    Let $(\Bsimple, \M)$ be an {\regular} paradox.
    \begin{enumerate}
        \item For any $z_0, z_1 \in \Z_2$, we have $\bigoplus_j r(j,0,z_0) \neq \bigoplus_j r(j,1,z_1)$.
        \item For each $l \in \Z_2$,  $\bigoplus_j r(j,l,0) = \bigoplus_j r(j,l,1)$.
    \end{enumerate}
\end{corollary}

The following lemma establishes necessary and sufficient conditions on the third-qubit measurements for an {\regular} quantum scenario $(\Bsimple, \M)$ to be a paradox.  

\begin{lemma}\label{lem:even_odd_multiples}
    Let $(\Bsimple, \M)$ be an {\regular} quantum scenario. It is a paradox if and only if $\delta(\lambda, C_0)$ and $\delta(\lambda, C_1)$ are \emph{even} multiples of $\pioverN$ and $\beta(\lambda, C_0) - \beta(\lambda, C_1)$ is an \emph{odd} multiple of $\pioverN$.
\end{lemma}
\begin{proof}
    Suppose that $(\Bsimple, \M)$ is a paradox. We follow the proof of Lemma \ref{lem:delta_beta_restrict}, but in this case we can use Corollary \ref{cor:opposite_parities} to deduce further restrictions. Let $l \in \Z_2$. After summing the $N$ equations of Eq.\ \ref{eq:eq1} and then subtracting the $N$ equations of Eq.\ \ref{eq:eq2}, the coefficient of $\pi$ on the RHS is
    \begin{equation}
        \bigoplus_{j=0}^{N-1} \left(a_j \oplus b_{K(j,l,0)}\right) \oplus \bigoplus_{j=0}^{N-1} \left(a_j \oplus b_{K(j,l,1)}\right) = \bigoplus_{j = 0}^{N-1} (1 \oplus r(j,l,0)) \oplus \bigoplus_{j = 0}^{N-1} (1 \oplus r(j,l,1)) = 0,
    \end{equation}
    where the last equality is by Corollary \ref{cor:opposite_parities}.2. Thus, we in fact have $N\delta(\lambda, C_l) \equiv 0 $ for each $l \in \Z_2$. Similarly, we can use Corollary \ref{cor:opposite_parities}.1 to deduce that $N(\beta(\lambda, C_0) - \beta(\lambda, C_1)) \equiv \pi $. Dividing each of these equations by $N$ gives us the desired result.

    Conversely, we can reverse the argument above to deduce that, for any $z_0, z_1 \in \Z_2$, the linear system $\Psi(z_0, z_1)$ satisfies $\bigoplus_j r(j,0,z_0) \oplus \bigoplus_j r(j,1,z_1) = 1$. By Lemma \ref{lem:opposite_parities}, $\Psi(z_0, z_1)$ is inconsistent, and Lemma~\ref{lem:inconsistent} then implies that $(\Bsimple, \M)$ is a paradox.
\end{proof}

Therefore, an {\regular} quantum scenario $(\Bsimple, \M)$ is a paradox if and only if the third-qubit measurements satisfy certain conditions. Below we list multiple families of quantum scenarios that satisfy these conditions and in the next subsection, we prove that these exhaust all {\regular} paradoxes. 

\begin{theorem}\label{thm:new_paradoxes}
    Let $N \in \N$, $N \geq 2$, and let $M_3 = \{C_0, C_1\}$. Then, $(\Bsimple, \M)$ is a paradox whenever the parameters $\lambda, C_0, C_1, M_1, M_2$ satisfy one of the following conditions:
    \begin{enumerate}[label=(\alph*)]
    \item $N=2$: The parameters describe the standard GHZ paradox.
    
    \item $N >2$, $N$ is even: There exist $s,t \in \Z$, with $s$ even and $t$ odd, satisfying $1 \leq s < N$ and $1 \leq t \leq s-1$, such that \begin{equation}\label{eqn:case_b}
        C_0 = 0, \quad \delta(\lambda, C_1) \equiv \pioverN s , \quad \beta(\lambda, C_1) \equiv -\pioverN t.
    \end{equation}
    Additionally, $M_1 = M_2 = \pioverN \Z_N$.
    In particular, if $t = s/2$, then $\lambda = \piovertwo - \pioverN t$ and $C_1 = \piovertwo$.
    
    \item There exist $s,t' \in \Z$, with $s$ even, satisfying $1 \leq s < N$ and $0 \leq t' \leq s/2-1$, such that
    \begin{equation}
        \begin{gathered}
            C_1 = \pi - C_0,\quad \delta(\lambda, C_0) \equiv \delta(\lambda, C_1) \equiv \pioverN s , \\ 
            \beta(\lambda, C_0) \equiv -\pioverN(t' + \tfrac{1}{2}) ,\quad \beta(\lambda, C_1) \equiv -\pioverN(s - t' - \tfrac{1}{2}).
        \end{gathered}
    \end{equation}
    Additionally, $M_1 = \pioverN \Z_N$, and $M_2 = \pioverN(\Z_N + \frac{1}{2})$.
    
    \item There exist $s,s',t' \in \Z$, with $s,s'$ even and $t'$ odd, satisfying $1 \leq s < s' < N$ and $1 \leq |t'| \leq N-1$, such that
    \begin{gather}
        \delta(\lambda, C_0) \equiv \pioverN s ,\quad \delta(\lambda, C_1) \equiv \pioverN s' ,\quad \beta(\lambda, C_1) - \beta(\lambda, C_0) \equiv \pioverN t'. \label{eqn:icky3}
    \end{gather}
    Additionally, let $\beta(\lambda, C_0) \equiv -\pioverN t$ ($0 < t < s$) and $\nu := \lceil t \rceil - t$. Then $M_1 = \pioverN \Z_N$, and $M_2 = \pioverN(\Z_N + \nu)$.
    \end{enumerate}
\end{theorem}
\begin{proof}
    In each case, the quantum scenario $(\Bsimple, \M)$ is maximally impossible and $|M_1| = |M_2| = N$. Moreover, the conditions on $\delta$ and $\beta$ from Lemma~\ref{lem:even_odd_multiples} are satisfied, so for any $z_0, z_1 \in \Z_2$, the linear system $\Psi(z_0, z_1)$ satisfies $\bigoplus_j r(j,0,z_0) \oplus \bigoplus_j r(j,1,z_1) = 1$. Therefore, by Remark \ref{rem:slight-generalisation}, $(\Bsimple, \M)$ is a paradox. 
\end{proof}

Except for case (a) and the specific instance of case (b) where $t = s/2$, the existence of parameters $\lambda, C_0 , C_1$ satisfying the conditions (b)--(d) of the above theorem is nontrivial. The existence of such parameters is established in Theorem \ref{thm:partial_classification} in the following subsection.

These paradoxes can be visualised by picturing the measurement angles of $M_1, M_2$ and the values of $\beta(\lambda, C_l)$ for $l\in\Z_2$ as points distributed around a circle. For example, Figure \ref{fig:diagram1} visualises two different paradoxes satisfying condition (b) with $N=8$, $C_1 = \piovertwo$ and $t=s/2=1$, $3$. The measurements in $M_1$ and $M_2$ are evenly spread modulo $\pi$, with each connected by a dotted line to its diametric opposite, corresponding to flipping the outcome. In the third circle, each pair of values $\beta, \beta + \delta$, which are joined by dotted line segments of the same colour, is symmetric under reflection about the $x$-axis. For different $s$ and $t$, the spacing between these pairs and their $\beta$ values varies but remains an integer multiple of $\pioverN $. 
See Figure~\ref{fig:diagram2} for illustrations of the paradoxes in the other cases.

\begin{figure}\vspace{-0.75cm}
    \centering
    \input{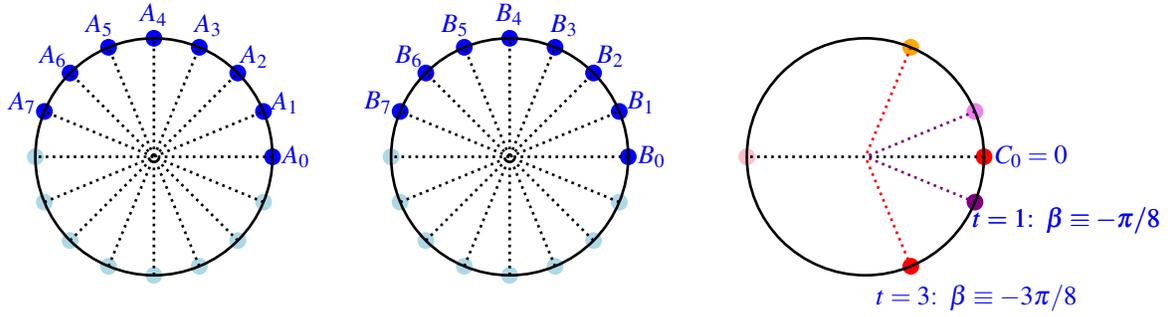}
    \vspace*{-1cm}\caption{Two paradoxes satisfying condition (b), where $N=8$, $C_1 = \piovertwo$, and $t=s/2$, for the cases $t = 1,3$.}
    \label{fig:diagram1}
\end{figure}

In case (b), when $t = s/2$, the paradoxes take the form $(\Bsimple<\lambda_{N,t}>, \M)$, where $\lambda_{N,t} = \piovertwo - \pioverN t$ for all even $N$ and odd $t$ satisfying $1 \leq t < N/2$. The third measurement set is $\{0,\frac{\pi}{2}\}$. This family is an extension of \cite[Theorem~8]{Abramsky2017}, which considered $\lambda_N = \piovertwo - \pioverN $ indexed by even $N\in \N$. Our construction, indexed by a tuple $(N,t)$, corresponds to a subset of the rationals and includes infinitely more paradoxes.

Remarkably, the other cases in Theorem \ref{thm:new_paradoxes} exhibit paradoxes from stranger empirical models. In each of these cases, although the resulting values of the functions $\delta$ and $\beta$ have relatively simple properties, the value of $\lambda$ and the measurement angles on the third qubit are complicated and appear to be very difficult to find analytically. It is as yet not even clear in these cases whether $\lambda \in \Q\pi$. Notably, the paradoxes in cases (c) and (d) do not involve the measurement angle $0$, i.e.\ the Pauli $X$, which is a distinguishing feature absent in all previously known examples.

\begin{theorem}\label{thm:inequivalent}
    Every paradox described in Theorem~\ref{thm:new_paradoxes} is distinct up to equivalence. 
\end{theorem}
The proof can be found in Appendix~\ref{AppendixB.1}.

\subsection{A partial classification of three-qubit nonlocality paradoxes}   
\label{section:4.classification}

We now derive a complete classification for all {\regular} paradoxes. If a paradox does not come from a maximal rank quantum scenario, then the associated linear systems may have a more complex structure that warrants further investigation. See Subsection \ref{section:4.complete_class} for further discussion.

Before presenting the main theorem, we first establish a few key lemmas. First, Lemma \ref{lem:first_two_qubit_mmts} is about the measurements on the first two qubits in any {\regular} paradox.

\begin{lemma}\label{lem:first_two_qubit_mmts}
    Let $(\Bsimple, \M)$ be an {\regular} quantum scenario. Then the $N$ measurement angles in both $M_1$ and $M_2$ are evenly spaced modulo $\pi$, with adjacent angles differing by $\pioverN $. Mathematically, we have (up to equivalence) $M_1 = \pioverN\Z_N$ and $M_2 = \pioverN\Z_N + \mu$ for some $\mu \in [0, \pi)$.
\end{lemma}
\begin{proof}
    For each $j \in \Z_N$, consider the two equations
    \begin{align}
        (A_j + B_{K(j,l,0)}) &\equiv \beta(\lambda, C_l) + (1 \oplus a_j \oplus b_{K(j,l,0)})\pi \quad \text{for } l \in \Z_2,\label{eqn:imposs_2}
    \end{align}
    which determine the impossible events on the contexts $(A_j, B_{K(j,l,0)}, C_l)$.
    Subtracting the $l=1$ equation from the $l=0$ equation gives us $B_{K(j,0,0)} - B_{K(j,1,0)} \equiv \beta(\lambda, C_0) - \beta(\lambda, C_1) \mod{\pi}$, i.e.\ the measurements $B_{K(j,0,0)}$ and $B_{K(j,1,0)}$ differ by a multiple of $\pioverN $. For each $j$, Eq.\ \ref{eqn:imposs_2} gives us a pair of measurements $\{B_{K(j,0,0)}, B_{K(j,1,0)}\}$. Since $(\Bsimple, \M)$ is maximally impossible, each measurement appears in exactly two such pairs. By the maximal rank property, it is possible to order these pairs to form a chain $\{B_{k_0}, B_{k_1}\},\ \{B_{k_{1}}, B_{k_{2}}\}, \ \ldots,\ \{B_{k_{N-2}}, B_{k_{N-1}}\},\ \{B_{k_{N-1}}, B_{k_0}\}$ consisting of all $N$ measurements on the second qubit. By taking linear combinations of the difference equations associated with these pairs, we deduce that all second-qubit measurement angles differ by multiples of $\pioverN $. Since there are $N$ measurements, the result for $M_2$ follows. The same result holds for $M_1$ by symmetry.
\end{proof}

The next two lemmas are used to show that the four cases (a)--(d) in Theorem \ref{thm:new_paradoxes} exhaust all $N$-regular paradoxes. Note that, in cases (a) and (b), one of the third-qubit measurements is $X$ ($C_0 = 0$); then a necessary and sufficient condition for a paradox is that $\lambda$ and the second measurement $C_1$ be such that $\delta(\lambda, C_1)$ is an even multiple of $\pioverN $ and $\beta(\lambda, C_1)$ \emph{itself} is an odd multiple of $\pioverN $ (since $\beta(\lambda, C_0) = 0$). Next, assuming that neither third-qubit measurement is $X$, in case (c) we have that $\delta(\lambda, C_l)$ is the same for each $l \in \Z_2$. To find all paradoxes with this property, we show in Lemma \ref{lem:death_by_algebra} that, given $\pioverN s$, with $s$ even, and $\lambda \in \left[\piovertwo - \frac{\pi}{2N}s, \piovertwo\right)$, there are at most two different measurement angles $C_0, C_1$ that satisfy $\delta(\lambda, C_l) \equiv \pioverN s$, and furthermore $C_1 = \pi - C_0$. Then, by imposing the requirement that $\beta(\lambda, C_0) - \beta(\lambda, C_1)$ be an odd multiple of $\pioverN $, we can completely classify all the $N$-regular paradoxes with equal values of $\delta(\lambda, C_l)$.

\begin{lemma}
    For $\lambda \in \piovertwointerval$ and $C \in [0, \pi)$, the equations $\delta (\lambda, C) \equiv \pi$ and $\beta(\lambda, C) \equiv 0$ are satisfied simultaneously if and only if either $\lambda = 0$ or $C=0$.  
\end{lemma}
\begin{proof}
    Follows from the properties of the functions $\beta$ and $\delta$. 
\end{proof}

\begin{lemma}\label{lem:death_by_algebra}
    Given $N,s \in \N$, $N > 2$, $1 \leq s < N$, and $t \in \R$, $0 < t < s$, there exist unique $\lambda = \lambda(s,t) \in \left[\piovertwo - \frac{\pi}{2N}s, \frac{\pi}{2}\right)$ and $C = C(s,t) \in (0, \pi)$ such that 
    \begin{equation}\label{eqn:pi_over_N_multiples}
        \delta(\lambda, C) \equiv \pioverN s \quad \text{and} \quad\beta(\lambda, C) \equiv -\pioverN t .
    \end{equation}
    Furthermore,
    \begin{enumerate}[label=(\roman*)]
        \item $\lambda (s,\frac{s}{2}) = \piovertwo - \frac{\pi}{2N}s$ and $C(s,\frac{s}{2}) = \frac{\pi}{2}$.

        \item For a fixed $s$, $\lambda(s,-)$ is two-to-one on $(0, s)$, with $\lambda (s, s-t) = \lambda (s,t)$.
        
        \item For a fixed $s$, $C(s,-)$ is bijective on $(0, s)$, with $C(s, s-t) = \pi - C(s,t)$.
    \end{enumerate}
\end{lemma}
\begin{proof}
    % We have that $s = N \iff \delta(\lambda, C) \equiv \pi  \iff C = 0 \iff \beta(\lambda, C) \equiv 0 $. This proves (b).
    For $s < N$, we note that for a fixed $\lambda$, $\delta(\lambda, -)$ has a minimum (modulo $2\pi$) of $\pi - 2\lambda$ at $C = \piovertwo$. Thus Eq.\ \ref{eqn:pi_over_N_multiples} is satisfied only if $\lambda \geq \piovertwo - \frac{\pi}{2N}s$.

    Let $\lambda := \lambda(s,t)$ and $C:= C(s,t)$. Then,  
    \begin{align}
        \delta(\lambda, C) \equiv \pioverN s
        \quad\iff\quad \pi - 2\arctan(\sin C\tan\lambda) \equiv \pioverN s \quad\iff\quad \sin C = \frac{\tan(\piovertwo - \frac{\pi}{2N}s)}{\tan\lambda}.\label{eqn:icky1}
    \end{align}
    Let us set $C^0(s, \lambda) := \arcsin\left(\tan(\piovertwo - \frac{\pi}{2N}s)/\tan\lambda\right)$ and $C^1 := \pi - C^0$, noting that $C^0$ and $C^1$ both satisfy Eq.\ \ref{eqn:icky1} for $C$. For a fixed $s$, we have that $C^0(s, -)$ and $C^1(s, -)$ are functions of $\lambda$ that give us the at most two different angles $C$ such that $\delta(\lambda, C) \equiv \pioverN s$. Substituting $C^0$ into the expression for $\beta$ (Eq.\ \ref{eqn:beta}), we get
    \begin{equation}\label{eqn:icky2}
        \beta\bigl(\lambda,\, C^0(s,\lambda)\bigr) = \arcsin\left(\cot\left(\frac{\pi}{2N}s\right) \cot\lambda\right) - 2 \arctan\left(\frac{\cos\frac{\lambda}{2} \cot\left(\frac{\pi}{2N}s\right) \cot\lambda}{\sqrt{1 -\cot^2\left(\frac{\pi}{2N}s\right) \cot^2\lambda}\cos\frac{\lambda}{2} + \sin\frac{\lambda}{2}}\right).
    \end{equation}
    
    Let $t^0(s, \lambda) := -\frac{N}{\pi} \beta(\lambda, C^0)$. By Lemma \ref{lem:t0-properties} below, $t^0(s, -)$ has a well-defined inverse $\lambda^0(s, -): (0, \frac{s}{2}] \rightarrow \left[\frac{\pi}{2} - \frac{\pi}{2N}s, \frac{\pi}{2}\right)$. Furthermore, from the definition of $\beta$, we have $\beta(\lambda, \pi+C^0) \equiv - \beta(\lambda, \pi-C^0)$. Thus 
    \begin{equation}\label{eq:beta-equidistant}
        \beta(\lambda, C^0) + \beta(\lambda, C^1) \equiv - \delta (\lambda, C^0) \equiv -\pioverN s ,
    \end{equation}
    so $\beta(\lambda, C^1) \equiv -\pioverN (s-t)$ if and only if $\beta(\lambda, C^0) \equiv -\pioverN t$.

    Putting all this together, we can construct $\lambda$ and $C$ as functions of $s$ and $t$ as follows:
    \begin{equation}
        \lambda(s,t) = \begin{cases}
            \lambda^0(s,t) \quad &\text{if } 0 < t \leq s/2,\\
            \lambda^0(s, s-t) \quad &\text{if } s/2 < t < s;
        \end{cases}
    \end{equation}
    \begin{equation}
        C(s,t) = \begin{cases}
            \arcsin\left(\tan(\piovertwo - \frac{\pi}{2N}s)/\tan\bigl(\lambda^0(s,t)\bigr)\right) \quad &\text{if } 0 < t \leq s/2,\\
            \pi - \arcsin\left(\tan(\piovertwo - \frac{\pi}{2N}s)/\tan\bigl(\lambda^0(s, s-t)\bigr)\right) \quad &\text{if } s/2 < t < s.
        \end{cases}
    \end{equation}
    The above reasoning shows that these are well-defined. Hence we have proved the statement of the lemma. Statement (i) follows by noting that $\delta(\piovertwo - \frac{\pi}{2N}s, \piovertwo) \equiv \pioverN s $ and $\beta(\piovertwo - \frac{\pi}{2N}s, \piovertwo) \equiv -\frac{\pi}{2N}s$. Statements (ii) and (iii) immediately follow from the constructions of $\lambda(s,t)$ and $C(s,t)$.
\end{proof}

\begin{lemma}\label{lem:t0-properties}
    Let $t^0(s, \lambda) := -\frac{N}{\pi} \beta(\lambda, C^0)$. The function $t^0(s, -)$ is strictly decreasing, and hence bijective, on $\left[\frac{\pi}{2} - \frac{\pi}{2N}s, \frac{\pi}{2}\right)$. Furthermore, $t^0\left(s, \left[\frac{\pi}{2} - \frac{\pi}{2N}s, \frac{\pi}{2}\right)\right) = (0, \frac{s}{2}]$.
\end{lemma}
\begin{proof}
    We have $\beta(\piovertwo - \frac{\pi}{2N}s, \piovertwo) = -\frac{\pi}{2N}s$, so $t^0(s, \piovertwo - \frac{\pi}{2N}s) = \frac{s}{2}$. Next, the expression for $\beta(\lambda, C^0)$ in Eq.\ \ref{eqn:icky2} shows that $t^0(s, -)$ has no discontinuities for $\lambda \in \left[\frac{\pi}{2} - \frac{\pi}{2N}s, \frac{\pi}{2}\right)$, and that $t^0(s, \lambda) \rightarrow 0$ as $\lambda \rightarrow \piovertwo^-$. To show that $t^0(s,-)$ is strictly decreasing on the given interval, one can find the partial derivative of Eq.\ \ref{eqn:icky2} with respect to $\lambda$ and derive that
    \begin{equation}
        \frac{\partial \beta(\lambda, C^0)}{\partial\lambda} = 0 \quad\iff\quad \sqrt{\frac{\cos^2\left(\frac{\pi}{2N}s\right) - \sin^2\lambda}{\left(\cos^2\left(\frac{\pi}{2N}s\right) - 1\right) \sin^2\lambda}} = -\frac{1}{\sin \lambda},
    \end{equation}
    which is never the case since $-1/\sin\lambda < 0$ for $\lambda \in \left[\frac{\pi}{2} - \frac{\pi}{2N}s, \frac{\pi}{2}\right)$. Thus, $t^0(s, -)$ is strictly monotone on the given interval.
\end{proof}

\begin{remark}\label{rem:no-t-greater-than-s}
    By the analysis in the two proofs above of the image of $t^0$ and the relationship between $\beta(\lambda, C^0)$ and $\beta(\lambda, C^1)$, we deduce that Eq.\ \ref{eqn:pi_over_N_multiples} is not satisfied if $s \leq t < N$.
\end{remark}

Figure \ref{fig:lambda_and_C} provides a visualisation of these technical lemmas, showing how the value of $t$ relative to $s$ affects the values of $\lambda$ and $C$ that satisfy Eq.\ \ref{eqn:pi_over_N_multiples}.

\begin{figure}
    \centering
    \resizebox{0.5\textwidth}{!}{\input{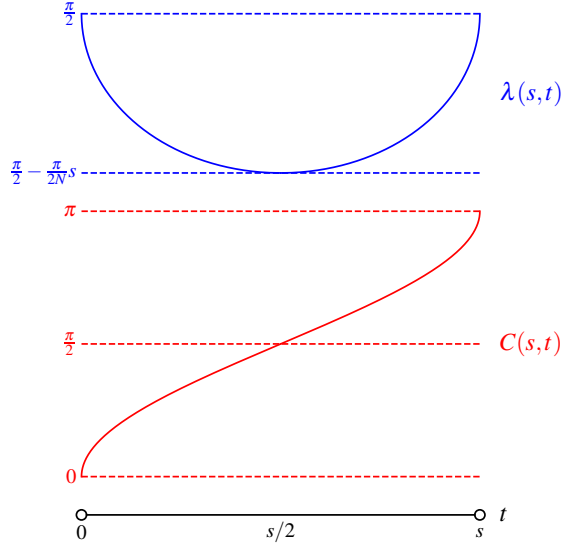}}
    \caption{A visualisation of Lemma \ref{lem:death_by_algebra}. For a fixed integer $s$ with $1 \leq s < N$, $t$ can take any real value greater than $0$ and less than $s$. The two graphs above the interval line for $t$ show the unique values of $\lambda$ and $C$ satisfying Eq.\ \ref{eqn:pi_over_N_multiples} for each $t$. Note the even symmetry of $\lambda(s,t)$ and the odd symmetry of $C(s,t)$, both about $t=s/2$.}
    \label{fig:lambda_and_C}
\end{figure}

We are now ready to prove the main theorem classifying all {\regular} paradoxes.

\begin{theorem}\label{thm:partial_classification}
    Let $(\Bsimple, \M)$ be an {\regular} paradox.  Then it is equivalent to one of the paradoxes described in Theorem \ref{thm:new_paradoxes}. Furthermore:
    \begin{enumerate}[label=(\roman*)]
        \item Condition (b) is satisfied by a unique {\regular} paradox for each $(s,t) \in \Z^2$, with $s$ even and $t$ odd, satisfying $1 \leq s < N$ and $1 \leq t \leq s - 1$. Condition (c) is satisfied by a unique {\regular} paradox for each $(s,t') \in \Z^2$, with $s$ even, satisfying $1 \leq s < N$ and $0 \leq t' \leq s/2 - 1$.
        \item Condition (d) is satisfied by at least one {\regular} paradox for some $N$ and $s,s',t' \in \Z$, with $s,s'$ even and $t'$ odd, satisfying $1 \leq s < s' < N$ and $1 \leq |t'| \leq N-1$.
    \end{enumerate}
\end{theorem}
\begin{proof}
By Lemma \ref{lem:even_odd_multiples}, $\delta(\lambda, C_l)$ is an even multiple of $\pioverN $ and $\beta(\lambda, C_0) - \beta(\lambda, C_1)$ is an odd multiple of $\pioverN $. We analyse four different cases that exhaust all possibilities for $N$, $C_0$, $C_1$, and $\delta(\lambda, C_l)$. In each case, the sets $M_1, M_2$ are determined by Lemma \ref{lem:first_two_qubit_mmts} and the requirement to satisfy the model's impossible event equations.

Suppose that $N=2$. Then $\delta(\lambda, C_l) \equiv \pi$ for both $l=0,1$, which is the case if and only if $\lambda = 0$. Then, since $\beta(0, C_0) - \beta(0, C_1) \equiv \piovertwo$, we have $C_1 - C_0 = \piovertwo$. Hence the paradox satisfies condition (a) in Theorem \ref{thm:new_paradoxes}.

Now suppose that $N > 2$. First, we check the case where $C_0 = 0$. Note that $\delta(\lambda, 0) \equiv \pi$, which is an even multiple of $\pioverN $ if and only if $N$ is even. Furthermore, $\beta(\lambda, 0) = 0$ for all $\lambda$. Consequently, any $C_1$ that satisfies Eq.\ \ref{eqn:case_b} will ensure that the difference $\beta(\lambda, C_0) - \beta(\lambda, C_1)$ is an odd multiple of $\pioverN $. By Lemma \ref{lem:death_by_algebra}, there exist unique satisfying $\lambda$ and $C_1$ for each $(s,t) \in \Z^2$, with $s$ even and $t$ odd, satisfying $1 \leq s < N$ and $1 \leq t \leq s - 1$. This paradox satisfies condition (b). By Remark \ref{rem:no-t-greater-than-s}, we cannot have a paradox if $s < t < N$.

Next, consider when both $C_0$ and $C_1$ are nonzero and $\lambda$, $C_0$ and $C_1$ are such that $\delta(\lambda, C_0) \equiv \delta(\lambda, C_1) \equiv \pioverN s$ for $1 \leq s < N$ and $s$ even. We additionally need the $\beta(\lambda, C_l)$'s to differ by an odd multiple of $\pioverN $. Let $t$ be such that $\beta(\lambda, C_0) \equiv -\pioverN t$. Lemma \ref{lem:death_by_algebra}, in particular statements (ii) and (iii), and the paragraph of the proof containing Eq.\ \ref{eq:beta-equidistant}, tell us that $C_1 = \pi - C_0$ and $\beta(\lambda, C_1) \equiv -\pioverN (s-t)$. Hence $\beta(\lambda, C_0) - \beta(\lambda, C_1) \equiv \pioverN (s - 2t)$. Since $s - 2t$ is odd if and only if $2t$ is odd, $t$ must be a half-integer. Thus this paradox satisfies condition (c). The bounds on $t'$ follow from Remark \ref{rem:no-t-greater-than-s}.

Finally, consider when both $C_0$ and $C_1$ are nonzero and $\delta(\lambda, C_0) \not\equiv \delta(\lambda, C_1)$. Referring to Eq.\ \ref{eqn:icky1}, a triple $(\lambda, C_0, C_1)$ satisfying condition (d) exists only if $C_0 = \arcsin\left(\tan(\piovertwo - \frac{\pi}{2N}s)/\tan\lambda\right)$ or $\pi - \arcsin\left(\tan(\piovertwo - \frac{\pi}{2N}s)/\tan\lambda\right)$, and similarly for $C_1$ (replacing $s$ with $s'$). Substituting these into $\beta(\lambda, C_0) - \beta(\lambda, C_1)$, we obtain a very complicated function that is difficult to analyse. However, solutions to Eq.\ \ref{eqn:icky3} with these substitutions do exist for some values of $N$, $s$, $s'$ and $t'$\footnote{E.g.\ if $N=8$, $s=4$, $s'=6$ and $t' = 1$, one can numerically find the solution $\lambda \approx 0.8129$, $C_0 \approx 1.2418$, $C_1 \approx 0.4028$.}.
\end{proof}

\subsection{A path to a complete classification of three-qubit nonlocality paradoxes}
\label{section:4.complete_class}
In the previous subsections, we detailed a comprehensive investigation of paradoxes that are {\regular}, i.e.\ maximally impossible, maximal rank paradoxes with two third-qubit measurements. The next step for the classification program is to drop some or all of these assumptions.

We observe that there exist maximally impossible paradoxes with $|M_3| = 2$ that are not maximal rank. For example, let $N = 4$. There exist $\lambda, C_0, C_1$ such that $\delta(\lambda, C_l) \equiv \frac{3\pi}{4} $ for $l \in \Z_2$, $\beta(\lambda, C_0) \equiv -\frac{\pi}{4} $, and $\beta(\lambda, C_1) \equiv -\frac{\pi}{2} $\footnote{Numerically, $\lambda \approx 0.4271$, $C_0 = \pi - C_1 \approx 1.1437$.}. Note that here $\delta(\lambda, C_l)$ is an odd multiple of $\pioverN $. As usual, let $M_1 = M_2 = \frac{\pi}{4}\Z_4$. See Figure~\ref{fig:diagram2} for an illustration. In this case, the four $\Z_2$-linear systems $\Psi(z_0, z_1)$ are all inconsistent, so that $(\Bsimple, \M)$ is a paradox. However, the parities $\bigoplus_j r(j,l,0)$ and $\bigoplus_j r(j,l,1)$ are not equal for either $l \in \Z_2$. Thus, a parity argument akin to the proof of Lemma \ref{lem:opposite_parities}, which considers all equations, is not applicable to two of the systems, which necessitate a parity argument using a proper subset of the equations. Further study is needed to determine the exact conditions for maximally impossible quantum scenarios to be paradoxes.

We have shown that there exist paradoxes $(\Bsimple<\lambda_{N,t}>, \M)$ where $\lambda = \piovertwo - \pioverN t$, $1 \leq t < N/2$, for all $N$ even and $t$ odd. However, there do not exist paradoxes $(\Bsimple<\lambda_{N,t}>, \M)$ such that $|M_3| = 2$ and $\lambda = \piovertwo - \pioverN t$ for $N$ \emph{odd}. It remains to be determined whether a paradox, with such a $\lambda$ or otherwise, is possible if we expand the measurement set $M_3$ to have more than two measurements.

Let us now examine general balanced states $\Bstate$ and seek necessary conditions on the $\lambda_i$'s for $(\Bstate, \M)$ to be a paradox. A key observation from our analysis of {\twolambdaszero} states is that the number of impossible events is maximised in each paradox that we have found. If an empirical model does not have many impossible events across its contexts, the likelihood of a consistent global assignment intuitively seems higher. So, we expect that for a paradox to be witnessed, the empirical model must contain a sufficiently large number of impossible events. 

The following lemma characterises empirical models where there are at least three contexts each containing four impossible events.
\begin{lemma}\label{lem:3context-4imposs_}
    Let $(\Bstate, \M)$ be a quantum scenario such that at least three contexts contain four impossible events each. Then $\Bstate$ must be an {\twolambdaszero} state. 
\end{lemma}
The proof can be found in Appendix~\ref{AppendixB.2}.

Along with Lemma \ref{lem:3context-4imposs_}, we observe other facts that also greatly restrict $\vec\lambda$ as long as we have sufficiently many impossible events across different contexts. For example, suppose $(\Bstate, \M)$ is a quantum scenario in which a measurement $A$ is in two contexts that both contain at least two impossible events, related by flipping the outcome of $A$ and another measurement on the same qubit, say $i$. Then $\lambda_i = 0$.

This evidence strongly suggests that two of the $\lambda_i$'s must be zero for $(\Bstate, \M)$ to be a paradox. This leads us to propose the following conjecture.

\begin{conjecture}\label{conj:two-lambda-zero}
    For the quantum scenario $(\Bstate, \M)$ to be a paradox, $\Bstate$ must be an {\twolambdaszero} state.
\end{conjecture}

\section*{Acknowledgments}

ND acknowledges support from the Canada Research Chair program, NSERC Discovery Grant RGPIN-2022-03103, the NSERC-European Commission project FoQaCiA, and the Faculty of Science of Simon Fraser University. SJ acknowledges support from the NSERC-European Commission project FoQaCiA.

\bibliographystyle{eptcs}
\bibliography{refs}

\newpage
\appendix
\section{Illustrations of the paradoxes}
\label{AppendixA}

\begin{figure}[!h]
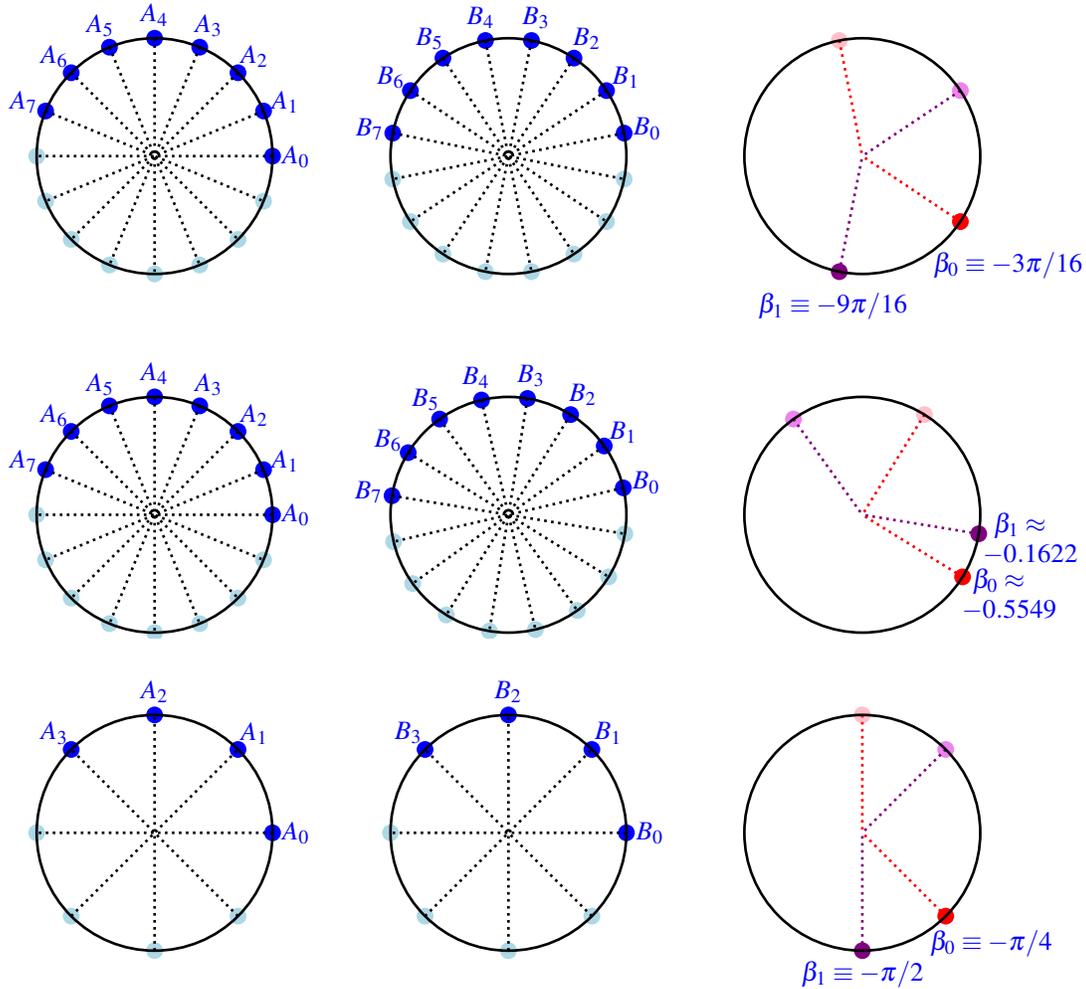

    \centering
    \resizebox{\textwidth}{!}{\input{media/case_c.pgf}}
    \resizebox{\textwidth}{!}{\input{media/case_d.pgf}}
    \resizebox{\textwidth}{!}{\input{media/non_max_rank.pgf}}
    \caption{From top to bottom: a paradox satisfying condition (c) of Theorem \ref{thm:new_paradoxes}, where $N$ = 8, $s = 6$, and $t' = 1$; a paradox satisfying condition (d), where $N = 8$, $s = 4$, $s' = 6$, and $t' = 1$; a non-{\regular} paradox satisfying $N=4$, $\delta (\lambda, C_l) \equiv \frac{3\pi}{4}$ for $l \in \Z_2$, $\beta(\lambda, C_0) \equiv -\frac{\pi}{4} $, and $\beta(\lambda, C_1) \equiv -\frac{\pi}{2}$.}
    \label{fig:diagram2}
\end{figure}

We present some more graphical illustrations of the paradoxes. In all our examples $|M_1| = |M_2| = N$, and $M_3 = \{C_0, C_1\}$. The existence of impossible events within a given context $(A, B, C)$ depends on whether one of the following conditions holds: \begin{align}
    A+B &\equiv \beta(\lambda, C) \quad \text{or} \quad \beta(\lambda, C)+ \pi \label{eq:26}\\
    A+B &\equiv \beta(\lambda, C+\pi) \quad \text{or} \quad \beta(\lambda, C+\pi) + \pi \label{eq:27}.
\end{align} 
Intuitively, given a $C_l \in M_3$, many of these equations must hold for various choices of $(A, B)$ in order to witness a paradox. By identifying $\pioverN  \Z_{N}$ with $\Z_N$, these $\beta$ values can be interpreted as `ticks' of a clock, which are achievable by $A+B$ for many different pairs $(A,B)$. In \cite[Theorem~8]{Abramsky2017}, all the paradoxes had ticks of $\pm 1$ for $C_1 = \frac{\pi}{2}$. For the case $t=s/2$ in the family (b) of Theorem~\ref{thm:new_paradoxes}, the ticks are $\pm t$. Notably, in all these examples, the ticks remain symmetric about the $x$-axis (see Figure \ref{fig:diagram1}). Furthermore, for $C_0 = 0$, the corresponding ticks are precisely $0$ and $\pi$. When $t\neq s/2$ in the family (b) in Theorem~\ref{thm:new_paradoxes}, the ticks for $C_1$ become asymmetric, specifically $-t$ and $s-t$, although we still observe a sort of symmetry about the $x$-axis if we view the ticks as indistinguishable. The other families in Theorem~\ref{thm:new_paradoxes} exhibit even more intricate structures, as their tick values are nonzero for both measurements and may not even be rational multiples of $\pi$ (see Figure~\ref{fig:diagram2}). However, we note that the \emph{differences} between tick values (i.e.\ the angles between ticks) are still appropriate multiples of $\pioverN $.

\section{Proofs of select results}
\label{appendixB}

\subsection{Proof of Theorem~\ref{thm:inequivalent}}
\label{AppendixB.1}
Since equivalence is defined up to local unitaries and qubit permutations, and any qubit permutation preserves equivalence, we restrict our analysis to local unitaries.  
\begin{lemma}
\label{lem:equivalence}
    Let $M$ and $M'$ be sets of equatorial measurements with $|M|, |M'| \ge 2$. If $UMU^\dag = M' \mod \pi$, then $U$ must be of the form $X^{a} P_{\theta}$ (up to a global phase) for some $\theta \in [0,2\pi)$ and $a\in \Z_2$, where $P_{\theta}$ is the phase gate.  
\end{lemma}
\begin{proof}
    The map $Ad : SU(2) \rightarrow SO(3)$, given by $Ad (g)(V) = gVg^{-1}$ is a $2$-fold covering of $SU(2) \rightarrow SO(3)$ with $\mathrm{ker} (Ad) = \{ I,-I\}$. Therefore, $Ad$ induces an isomorphism $PU(2) \cong SO(3)$, which establishes that single-qubit unitaries, up to a global phase, correspond to rotations of the Bloch sphere. The rotation $R_{\vec{n}} (\theta)$ of the Bloch sphere about an arbitrary axis $\vec{n} = (n_x, n_y, n_z)$ is given by \begin{equation}
        R_{\vec{n}} (\theta) = \exp \left( -i\frac{\theta}{2}(n_x X + n_y Y + n_z Z) \right) . \end{equation} 
    An arbitrary Bloch sphere rotation can be expressed as $U = R_z (\theta_1) R_x (\gamma) R_z (\theta_2)$, where $\theta_1, \theta_2 \in [0,2\pi)$ and $\gamma \in [0,\pi]$. Importantly, $U$ maps any great circle on the Bloch sphere to another great circle, as it represents a rotation of the sphere. Let us denote the equator of the Bloch sphere by $S$. Note that $R_z (\theta) (S) = S$ for all $\theta$ and $R_x(\gamma) (S)= S$ only if $\gamma = 0$ or $\pi$. So, $U(S) \neq S$ when $\gamma \neq 0$ or $\pi$. Since $S$ and $U(S)$ are great circles, they intersect precisely at two points separated by $\pi$. This shows that for $U$ to map at least two equatorial measurements that differ by less than $\pi$ to equatorial measurements, $\gamma$ must be $0$ or $\pi$. Therefore $U$ must be of the form $R_z (\theta)$ or $R_z(\theta_1) R_x (\pi) R_z (\theta_2)$, which is equivalent (up to a global phase) to $X^a P_{\theta}$ for some $\theta \in [0,2\pi)$ and $a\in \Z_2$.
\end{proof}

\begin{lemma}\label{lem:equivalence2}
    Let $\lambda, \lambda' \neq 0$. If $(\Bsimple<\lambda>, \M)$ and $(\Bsimple<\lambda'>, \M')$ are equivalent, then $\lambda = \lambda'$. Moreover, the local unitary $U$ establishing the equivalence must be of the form \begin{equation}
        U=P_{\theta_1} \otimes P_{\theta_2} \otimes I \quad \text{or} \quad XP_{\theta_1} \otimes XP_{\theta_2} \otimes X
    \end{equation}
    such that $\theta_1 + \theta_2 \equiv 0$.
\end{lemma}
\begin{proof}
    Let $(\Bsimple<\lambda>, \M)$ and $(\Bsimple<\lambda'>, \M')$ be equivalent via the local unitary $U$. By Lemma~\ref{lem:equivalence}, $U$ must be of the form $e^{i\phi} \bigotimes_{j=1}^3 X^{a_j} P_{\theta_j}$ for some $\phi, \theta_j \in [0,2\pi)$, $a_j \in \Z_2$ for $j=1,2,3$. Also, we have that \begin{equation}\label{eq:lambda-lambda'}
        \Bsimple<\lambda'> = e^{i\phi} \bigotimes_{j=1}^3 X^{a_j} P_{\theta_j} \Bsimple<\lambda> .
    \end{equation} Expanding Eq.\ \ref{eq:lambda-lambda'}, we have the following:
    \begin{align}
        &\cos{\frac{\lambda'}{2}} \ket{000} + \sin{\frac{\lambda'}{2}} \ket{001} + \sin{\frac{\lambda'}{2}} \ket{110} + \cos{\frac{\lambda'}{2}} \ket{111} \nonumber\\
          & = \cos{\frac{\lambda}{2}} e^{i\phi}\ket{a_1 a_2 a_3} + \sin{\frac{\lambda}{2}} e^{i(\theta_3+\phi)} \ket{a_1 a_2 \overline{a}_3} + \sin{\frac{\lambda}{2}} e^{i(\theta_1 +\theta_2 +\phi)} \ket{\overline{a}_1 \overline{a}_2 a_3} + \cos{\frac{\lambda}{2}} e^{i(\theta + \phi)} \ket{\overline{a}_1 \overline{a}_2 \overline{a}_3}, \label{eq:expanded-eq}
    \end{align}
    where $\overline{a}_j = a_j \oplus 1$ and $\theta = \theta_1 +\theta_2 + \theta_3$. It is easy to see that for Eq.\ \ref{eq:expanded-eq} to hold, the only possible choices for $(a_1, a_2, a_3)$ are $\{ (0,0,0),\ (0,0,1),\ (1,1,0),\ (1,1,1) \}$. Since all coefficients are real on the LHS, that also must be the case on the RHS. So, $\phi = 0$ or $\pi$. All the coefficients on the RHS are either $\pm \cos{\frac{\lambda}{2}}$ or $\pm \sin{\frac{\lambda}{2}}$. Comparing the amplitudes of the basis vector $\ket{000}$ on both sides, we have \begin{equation} \label{eq:app_cos}
         \cos{\frac{\lambda'}{2}} = \pm \cos{\frac{\lambda}{2}} \quad \text{or} \quad \pm \sin{\frac{\lambda}{2}}.
    \end{equation} 
    In both cases of Eq.\ \ref{eq:app_cos}, the only solutions are $\lambda '= \lambda$ or $\lambda' = \pi - \lambda$. However, the second case is not feasible since $\lambda , \lambda' < \piovertwo$. This further restricts the choice for $(a_1, a_2, a_3)$ to either $(0,0,0)$ or $(1,1,1)$. For both choices, Eq.\ \ref{eq:expanded-eq} is satisfied only if the following conditions hold: \begin{equation}
        \phi \equiv 0, \quad \theta_3+\phi \equiv 0, \quad \theta_1 +\theta_2 +\phi \equiv 0, \quad \text{and} \quad \theta_1 +\theta_2 + \theta_3 + \phi \equiv 0.  
    \end{equation}
    This reduces $U$ to the desired form.
\end{proof}

\begin{proof}[\textit{Proof of Theorem~\ref{thm:inequivalent}}]
    The triples $(\lambda, C_0, C_1)$ corresponding to the paradoxes in Theorem~\ref{thm:new_paradoxes} are distinct for different paradoxes. By Lemma~\ref{lem:equivalence2}, equivalent paradoxes must have the same $\lambda$ value and third-qubit measurements can only be reflected about the $X$-axis.  So, the only possible scenario where $(\Bsimple<\lambda>, \M)$ and $(\Bsimple<\lambda'>, \M')$ are equivalent is if $\lambda = \lambda'$ and $C'_l$ is a reflection of $C_l$ about the $X$-axis for both $l=0,1$. This is not possible since all measurements are in $[0,\pi)$. Hence, distinct triples $(\lambda, C_0, C_1)$ correspond to inequivalent paradoxes.
\end{proof}

\subsection{Proof of Lemma~\ref{lem:3context-4imposs_}}
\label{AppendixB.2}
\begin{proof}
    By Lemma~\ref{lem:impossible-events}.1, a context $(A,B,C)$ has 4 impossible events if and only if $\delta(\lambda_1, A) \equiv \delta(\lambda_2, B) \equiv \delta(\lambda_3, C) \equiv \pi$. Moreover, $\delta(\lambda, \varphi) \equiv \pi $ if and only if $\lambda=0$ or $\varphi = 0$. 
    
    If one of the three contexts is such that all measurement angles are nonzero, i.e.\ $A, B, C \neq 0$, then it follows that $\lambda_i = 0$ for all $i = 1,2,3$, implying that $\Bstate$ is the $\mathrm{GHZ}$ state.
    
    Alternatively, if a context has two nonzero measurements, e.g.\ $(A,B,0)$ with $A,B \neq 0$, then we have $\lambda_1 = \lambda_2 = 0$. When one of the contexts is $(0,0,0)$, it is impossible to derive any conclusions regarding the $\lambda_i$'s. However, if the remaining two contexts each contain exactly one nonzero measurement, then the nonzero measurements must be on different qubits, such as $(A,0,0)$ and $(0,B,0)$. In this case, we also have $\lambda_1 = \lambda_2 = 0$. Therefore $\Bstate$ is an {\twolambdaszero} state.
\end{proof}

\end{document}